\documentclass{elsarticle}
\usepackage{amssymb,latexsym,amsthm}
\usepackage{amsmath,amsfonts}
\usepackage{enumerate}
\usepackage{geometry}
\usepackage{url}
\ifdefined\directlua
\usepackage{fontspec}
\usepackage{microtype}
\fi
\usepackage[british]{babel}
\newcommand{\cS}{\mathcal{S}}

\newtheorem*{mth}{Main Theorem}

\newcommand{\GG}{\mathbb{G}}

\newcommand{\cC}{\mathcal{C}}

\newcommand{\cW}{\mathcal{W}}

\newcommand{\LL}{\mathbb{L}}

\newcommand{\cG}{\mathcal{G}}

\newcommand{\PG}{\mathrm{PG}}

\newcommand{\FF}{\mathbb{F}}

\newcommand{\Rad}{\mathrm{Rad}\,}

\newcommand{\codim}{\mathrm{codim}\,}
\theoremstyle{plain}
\newtheorem{lemma}{Lemma}[section]
\newtheorem{theorem}[lemma]{Theorem}

\newtheorem{proposition}[lemma]{Proposition}

\theoremstyle{definition}
\newtheorem{remark}[lemma]{Remark}

\def\nonsquare{\ensuremath{%
    \setbox0\hbox{$\square$}%
    \rlap{\hbox to \wd0{\hss\slash\hss}}\box0
}}
\begin{document}
\begin{frontmatter}
\title{Minimum distance of Symplectic Grassmann codes}
\author[IC]{Ilaria Cardinali\corref{cor1}} 
\ead{ilaria.cardinali@unisi.it}
\address[IC]{Department of Information Engineering, University of Siena,
Via Roma 56, I-53100, Siena, Italy}
\author[LG]{Luca Giuzzi}
\ead{luca.giuzzi@unibs.it}
\address[LG]{D.I.C.A.T.A.M., Section of Mathematics,
University of Brescia,
Via Branze 53, I-25123, Brescia, Italy}
\cortext[cor1]{Corresponding author}
\begin{abstract}
In this paper
we introduce  Symplectic Grassmann codes,
in analogy to ordinary Grassmann codes and
Orthogonal Grassmann codes, as projective codes
defined by symplectic Grassmannians.
Lagrangian--Grassmannian codes are a special class
of Symplectic Grassmann codes.
We describe all the parameters of line Symplectic Grassmann
codes and we provide the full
weight enumerator for the
Lagrangian--Grassmannian codes of rank $2$ and $3$.
\end{abstract}
\end{frontmatter}

\section{Introduction}
Grassmann codes have been introduced in \cite{R1,R2} as generalizations
of Reed--Muller codes of the first order;
they have been extensively investigated ever since.
Their parameters, as well as some of their higher weights have been fully
determined in \cite{No}. These are projective codes, arising from the
Pl\"ucker embedding of a $k$--Grassmannian. A further point of interest
is that the weight distribution provides some interesting insight on
the geometry of the embedding itself.

Codes arising from
the Pl\"ucker embedding of the $k$--Grassmannian of an orthogonal polar
space have been introduced in a recent series of papers \cite{IL0,ILP,IL1}.
 In \cite{IL0}, we computed the
minimum distance for the codes arising
from orthogonal dual polar spaces
of rank $2$ and $3$ and provided a general bound on the minimum
distance. More recently, in \cite{ILP}, for $q$ odd
the minimum distance for all Line Polar Grassmann codes of
orthogonal type has been determined.
In \cite{IL1} an encoding scheme,
as well as strategies for decoding and error correction, has been
proposed for Line Polar Grassmann codes. We point out that, even if the
parameters of the codes under consideration arise from the Grassmann
embedding, neither the encoding scheme we considered nor the error
correction strategy we proposed
make direct use of Pl\"ucker coordinates.

The aim of the present paper is to provide results analogous to those
of \cite{IL0,ILP} for codes arising from the Pl\"ucker embedding
of $k$--Grassmannians of symplectic type.

More in detail,
we shall denote by
$\cW(n,k)$, the projective code defined by the image under the
Pl\"ucker embedding of the $k$--symplectic Grassmannian $\Lambda_{n,k}$
defined by a non--degenerate
alternating bilinear form $\sigma$ on a vector space $V:=V(2n,q)$
of dimension $2n$ over a finite field $\FF_q$.
This will be referred as a \emph{Symplectic Grassmann
Code}.

The paper is organized as follows:
in Section \ref{PRE} some basic notions about projective codes and
symplectic Grassmannians are recalled;
Section \ref{LLG} is dedicated to Line Symplectic
Grassmann codes and contains our main results for $k=2$;
Section \ref{DPS} is dedicated to the case of rank $k=3$. Overall, in
these sections we prove the following.
\begin{mth}
  The code $\cW(n,k)$ has parameters
  \[ N=\prod_{i=0}^{k-1}(q^{2n-2i}-1)/(q^{i+1}-1),\qquad
     K={2n\choose k}-{2n\choose{k-2}}. \]
Furthermore,
     \begin{itemize}
     \item For $k=2$, its minimum distance is $q^{4n-5}-q^{2n-3}$;
     \item For $n=k=3$, its minimum distance is $q^6-q^4$.
     \end{itemize}
\end{mth}
Finally,
in Section~\ref{bounds} we discuss some
further bounds  for the minimum distance in the general case of Symplectic
Grassmann codes arising from higher weights of
Grassmann codes.

We point out that the code $\cW(n,n)$
where  $k=n$, corresponding to the so called
\emph{dual polar space},
has already been introduced under the name
of Lagrangian-Grassmannian code of rank $n$ in \cite{PZ}, where some
bounds on the parameters have been obtained.

\section{Preliminaries}
\label{PRE}
A $[N,K,d_{\min}]$ projective system $\Omega\subseteq\PG(K-1,q)$
is just a set of $N$ distinct points in $\PG(K-1,q)$ whose
span is $\PG(K-1,q)$ and  such that for any
hyperplane $\Sigma$ of $\PG(K-1,q)$
\[ \#(\Omega\setminus\Sigma)\geq d_{\min}. \]
It is well known that existence of a $[N,K,d_{\min}]$ projective system is
equivalent to that of a
projective linear code $\cC$ with the same parameters.
Indeed, several codes can be obtained by taking as generator matrix
$G$ the matrix whose columns are the coordinates of the points of $\Omega$
normalized in some way. As the order of the points, the choice of
coordinates as well as the normalization adopted change, we obtain potentially
different codes arising from $\Omega$, but all of these turn out to be
equivalent. As such, in the following discussion, they will be silently
identified and we shall write $\cC=\cC(\Omega)$.
The spectrum of the intersections of $\Omega$ with the
hyperplanes of $\PG(K-1,q)$ provides the list of the weights of $\cC$;
we refer to \cite{TVZ} for further details.

Let now and throughout the paper
$V:=V(2n,q)$ be a $2n$-dimensional vector space equipped
with a non--degenerate bilinear alternating form $\sigma$.
Denote by $\cG_{2n,k}$ the $k$-Grassmannian of the projective
space $\PG(V)$, that is the point--line
geometry whose points are the $k$-dimensional subspaces of $V$
and whose lines are the sets
\[ \ell_{W,T}:=\{ X: W\leq X\leq T, \dim X=k \} \]
with $\dim W=k-1$ and  $\dim T=k+1$.
A projective embedding of $\cG_{2n,k}$ is a function
$e:\cG_{2n,k}\to\PG(U)$ such that $\langle e(\cG_{2n,k})\rangle=\PG(U)$
and each line of $\cG_{2n,k}$ is mapped onto a line of $\PG(U)$.
The dimension of $U$ is called \emph{dimension of the embedding}.
It is well known that
the geometry $\cG_{2n,k}$ affords a projective
embedding $e_{k}^{gr}:\cG_{2n,k}\to\PG(\bigwedge^kV)$ by means
of Pl\"ucker coordinates. In particular, $e_k^{gr}$ maps an arbitrary
$k$--dimensional subspace $\langle v_1,v_2,\ldots,v_k\rangle$ of
$V$ to the point $\langle v_1\wedge v_2\wedge\cdots\wedge v_k\rangle$.
The image $e_{k}^{gr}(\cG_{2n,k})$ is a projective variety
of $\PG(\bigwedge^kV)$,
usually denoted by the symbol $\GG(2n-1,k-1)$,
see \cite[Lecture 6]{Ha}, called the \emph{Grassmann variety}.

The symplectic Grassmannian $\Lambda_{n,k}$ induced by $\sigma$,
is defined for $k=1,\ldots,n$ as the subgeometry of $\cG_{2n,k}$ having as
points the totally $\sigma$--isotropic
subspaces of $V$ of dimension $k$ and as lines
\begin{itemize}
\item for $k< n$, the sets of the form
   \[ \ell_{W,T}:=\{ X: W\leq X\leq T, \dim X=k \} \]
   with $T$ totally isotropic and $\dim W=k-1$, $\dim T=k+1$.
 \item for $k=n$, the sets of the form
   \[ \ell_{W}:=\{ X: W\leq X, \dim X=n\} \]
   with $\dim W=n-1$, $W$ totally isotropic.
\end{itemize}
For $k=n$, $\Lambda_{n,n}$ is usually called \emph{dual polar space of rank $n$}
or \emph{Lagrangian Grassmannian}.

The image of $\Lambda_{n,k}$ under the Pl\"ucker embedding
$e_k^{gr}$ is a  subvariety
$\LL(n-1,k-1)$ of the Grassmann variety
$\GG(2n-1,k-1)$.%

Let $\Sigma=\langle\LL(n-1,k-1)\rangle<\PG(\bigwedge^k V)$.
It is well known, see \cite{Ba,PS}, that
\[\dim \Sigma={2n\choose k}-{2n\choose{k-2}};\]
indeed, the variety $\LL(n-1,k-1)$ is the full intersection of
$\GG(2n-1,k-1)$ with a
suitable subspace of $\bigwedge^kV$ of codimension ${2n\choose{k-2}}$.

The following formula provides
the length of $\cW(n,k)$:
\begin{equation}
\label{e0a}
\#\LL(n-1,k-1)=\#\Lambda_{n,k}=\prod_{i=0}^{k-1}(q^{2n-2i}-1)/(q^{i+1}-1).
\end{equation}

As pointed out before, the pointset of
$\LL(n-1,k-1)$ is a projective system of $\PG(\Sigma)$; thus this
determines an associated projective code which we
shall denote by $\cW(n,k)$
and call it a \emph{Symplectic Grassmann code}.
When $n=k$  Symplectic Grassmann codes are known in the literature also
as \emph{Lagrangian-Grassmannian codes}.
A straightforward consequence of the remarks
presented above is the following lemma.
\begin{lemma}
  The code $\cW(n,k)$ has length $N=\#\LL(n-1,k-1)$ and
  dimension $K=\dim\Sigma$.
\end{lemma}
\section{Line Symplectic Grassmann Codes}
\label{LLG}
Throughout this section $\cS:=W(2n-1,q)$ denotes a non-degenerate
symplectic polar space
defined by the non--degenerate alternating bilinear form
$\sigma$ on $V$ of rank $n$.
By $\theta$ we shall denote a different (possibly degenerate)
alternating bilinear form.
We
shall also write $\perp_{\sigma}$ and $\perp_{\theta}$ for the orthogonality
relations induced by  $\sigma$ and $\theta$ respectively.
Finally, recall that the \emph{radical} of $\theta$ is the set
\[ \Rad\theta:=\{ x\in\cS : x^{\perp_{\theta}}=\PG(V) \}=
\{ x\in\cS: \forall y\in\cS, \theta(x,y)=0 \}. \]

For $k=2$ the expression \eqref{e0a} becomes
\[ N:=\#\Lambda_{n,2}=\frac{(q^{2n}-1)(q^{2n-2}-1)}{(q-1)(q^2-1)}; \qquad K:=
2n^2-n-1. \]

It is well known that any bilinear alternating form $\theta$
determines an hyperplane of $\bigwedge^2 V$ and conversely;
hence, the minimum distance of $\cW(n,2)$ can be deduced from
the maximum number of lines which are
simultaneously totally isotropic for both
the forms $\theta$ and $\sigma$, under the assumption $\theta\neq\sigma.$
In order to determine this number we follow an approach similar to
that of \cite[Lemma 3.2]{ILP}.
\begin{lemma}
Suppose $M$ is the matrix representing $\sigma$ and $S$
the matrix representing $\theta$ with respect to a given reference system.
For any $p\in \cS$ we have $p^{\perp_{\sigma}}\subseteq p^{\perp_{\theta}}$
if, and only if, $p$ is an eigenvector of the matrix $M^{-1}S$.
\end{lemma}
\begin{proof}
  Since $M$ is non--singular, if $M^{-1}Sp=\mathbf{0}$, then $Sp=0$, that
  is to say $p\in\Rad\theta$, i.e. $p^{\perp_{\theta}}=\PG(V)$.
  In this case, obviously,
  $p^{\perp_{\sigma}}\subseteq p^{\perp_{\theta}}$.
  \par
  When $p\not\in\Rad\theta$, we have $\dim p^{\perp_{\sigma}}=\dim p^{\perp_{\theta}}$.
  Thus, $p^{\perp_{\sigma}}\subseteq p^{\perp_{\theta}}$ if, and only if,
  $p^{\perp_{\sigma}}=p^{\perp_{\theta}}$, that is to say the systems of equations
  $x^TMp=0$ and $x^TSp=0$ are equivalent. This yields
  $Sp=\lambda Mp$ for some $\lambda\neq 0$, whence $p$ is an eigenvector
  of eigenvalue $\lambda$ for $M^{-1}S$.
\end{proof}
Write now
\[ N_0:=\#\{p\in \cS: p^{\perp_{\sigma}}\not\subseteq p^{\perp_{\theta}}
\},\qquad N_1:=\#\{p\in \cS: p^{\perp_{\sigma}}\subseteq p^{\perp_{\theta}}\}.
 \]
Clearly, $N_0=\frac{q^{2n}-1}{q-1}-N_1.$

For any $p\in \cS$, a line $\ell$ through $p$ is both totally
$\sigma$--isotropic and $\theta$--isotropic if, and only if,
$\ell\in p^{\perp_{\sigma}}\cap p^{\perp_{\theta}}$.
In particular,
\begin{itemize}
\item  if $p^{\perp_\sigma}\subseteq p^{\perp_{\theta}}$, then
  $\frac{q^{2n-2}-1}{q-1}$ lines through $p$ are both $\sigma$-- and $\theta$--isotropic;
\item if $p^{\perp_{\sigma}}\not\subseteq p^{\perp_{\theta}}$, then
  $p^{\perp_{\sigma}}\cap p^{\perp_{\theta}}$ is a subspace of
  codimension $2$ in $\PG(V)$ and
  the number
  of lines which are both $\sigma$-- and $\theta$--isotropic is
  $\frac{q^{2n-3}-1}{q-1}$.
\end{itemize}
Denote now by $\eta$ the number of lines of $\cS$ which are
simultaneously totally $\sigma$-- and
$\theta$--isotropic. As each line contains $(q+1)$ points, we have
\begin{equation}
\label{e1}
 (q+1)\eta = N_0\frac{q^{2n-3}-1}{q-1}+N_1\frac{q^{2n-2}-1}{q-1}=
q^{2n-3}N_1+\frac{(q^{2n}-1)(q^{2n-3}-1)}{(q-1)^2}.
\end{equation}
Clearly, $\eta$ is maximum when $N_1$ is maximum.
In the remainder of this section we shall determine exactly how large
$N_1$ can be.
\begin{lemma}\label{eigenvectors}
  If the matrix $M^{-1}S$ has
  has just two eigenspaces, one of dimension $2n-2$, the other of
  dimension $2$, then the number of eigenvectors of $M^{-1}S$ is
  maximum.
\end{lemma}
\begin{proof}
  In order for the number of eigenvectors of $M^{-1}S$ to be maximum we
  need $M^{-1}S$ to be diagonalizable. Suppose that there are at least
  $3$ distinct eigenspaces, say $V_{\alpha}$, $V_{\beta}$ and $V_{\gamma}$
  of dimensions respectively $a,b,c$ with $a\leq b\leq c$.
  Then, the number of eigenvectors is
  \[ \# V_{\alpha}+\# V_{\beta}+\# V_{\gamma}-3= q^{a}+q^b+q^{2n-a-b}-3. \]
  Clearly $a,b>1$ and $a+b<2n$. In particular, this number is maximum
  for $a=b=1$, in which case we get
  \[ \# V_{\alpha}+\# V_{\beta}+\# V_{\gamma}-3=2q+q^{2n-2}-3<q^2+q^{2n-2}-2.\]
  Thus, the maximum number of eigenvectors which can be obtained with just
  $2$ eigenspaces, say $V_{\alpha}$ and $V_{\beta}$ is larger than that
  possible with $3$ distinct eigenspaces.

  Observe that since $S$ is antisymmetric, its rank is necessarily even;
  in particular, the rank of $M^{-1}S$ is also even.
  Suppose now $\lambda\neq0$ to be an eigenvalue of $M^{-1}S$
  and suppose than the corresponding eigenspace is maximum and
  has dimension $g$.
  The number of simultaneously $\sigma$-- and $\theta$--isotropic
  lines $\ell=\langle v_1,v_2\rangle$
  is the same as the number of lines which are simultaneously
  $\sigma$-- and $(\theta-\lambda\sigma)$--isotropic,
  as $\sigma(v_1,v_2)=0$ and $\theta(v_1,v_2)=0$ yields
  $(\theta-\lambda\sigma)(v_1,v_2)=\theta(v_1,v_2)-\lambda\sigma(v_1,v_2)=0$.
  The latter alternating form,
  say $\theta'=\theta-\lambda\sigma$ is represented the matrix
  $S'=S-\lambda M$.
  In particular, we can replace $S$ with $S'$ and we get
  \[ M^{-1}S'=M^{-1}(S-\lambda M)=M^{-1}S-\lambda I. \]
  For this new matrix, $0$ is an eigenvalue with eigenspace of dimension $g$.
  Thus, $g$ must be even.

  We conclude that the maximum number of eigenvectors might occur
  when $g=2n-2$ and
  there is a further eigenspace of dimension $2$, that is
  \[ \# V_{\lambda}+\# V_{\mu}-2=q^{2n-2}+q^{2}-2. \]
\end{proof}

The case of the previous lemma corresponds to
\[ N_1=\frac{q^{2n-2}-1}{q-1}+\frac{q^2-1}{q-1}. \]
Plugging in this value in \eqref{e1} we obtain
\[ \eta=\frac{q^{4n-3}+q^{4n-4}-q^{4n-5}-q^{2n-1}-2q^{2n-2}+q^{2n-3}+1}{(q-1)(q^2-1)}. \]
whence we get the following lemma.
\begin{lemma}\label{lemma-min-dist}
\[ d_{\min}(\cW(n,2))\geq q^{4n-5}-q^{2n-3}. \]
\end{lemma}

We are now ready to prove our main theorem for Line Symplectic Grassmann
codes.
\begin{theorem}
\label{tt}
  The minimum distance of the code $\cW(n,2)$ is
  $d_{\min}(\cW(n,2))=q^{4n-5}-q^{2n-3}$.
\end{theorem}
\begin{proof}
We shall show that, given a non-degenerate alternating form $\sigma$
represented by a matrix $M$, it is always possible to define an
alternating form $\theta$ represented by a matrix $S$ such that
$M^{-1}S$ has only two eigenspaces, one of dimension $2n-2$ and the
other of dimension $2$.

In order to prove this, let $\ell=\langle v_1,v_2\rangle$
be a line of $\PG(2n-1,q)$ which
is not $\sigma$--isotropic and define an alternating form $\theta$ such that
$\theta(v_1,v_2)=\sigma(v_1,v_2)$ and
$\Rad\theta=\ell^{\perp_{\sigma}}$.

Take $B={B}_1\cup B_2$ to be an ordered basis of $V$ where $B_1=(v_1,v_2)$ and ${B}_2$ is an ordered basis of $\ell^{\perp_{\sigma}}$.

Let $M$ be the matrix representing $\sigma$ with respect to $B$. We can suppose $M=\operatorname{diag}(M_{11}, M_{22})$ to be a block diagonal matrix where
\[M_{11}=\left(\begin{array}{rr}
0&1\\
-1&0\\
\end{array}\right)\,\,{\rm and }\,\, M_{22}=\left(\begin{array}{rr}
{O}_n& {I}_n\\
-{I}_n&{O}_n\\
\end{array}\right)\]
with ${O}_n$ the null $(n\times  n)$--matrix and ${I}_n$
the $(n\times  n)$--identity matrix.
By construction, the matrix $S$ representing $\theta$ with respect to
${B}$ is also block diagonal
$S=\operatorname{diag}(S_{11}, {O}_{2n-2})$ with
$S_{11}=\left(\begin{array}{rr}
0&1\\
-1&0\\
\end{array}\right)$ and ${O}_{2n-2}$ the null
$(2n-2)\times (2n-2)$--matrix.

Hence $M^{-1}S$ is the block diagonal matrix
$M^{-1}S=\operatorname{diag}({I}_2, {O}_{2n-2})$.
Clearly, $M^{-1}S$ has only two eigenspaces, one of dimension $2n-2$ and
 the other of dimension $2.$
The thesis now
follows from Lemma~\ref{eigenvectors} and Lemma~\ref{lemma-min-dist}.

\end{proof}

The proof of Theorem \ref{tt} holds also for the code $\cW(2,2)$ arising
from the dual polar space $\Lambda_{2,2}$. However, in this case we can
easily provide the full weight enumerator.

\begin{proposition}
  The code $\cW(2,2)$ has exactly $3$ nonzero weights, namely $q^3+q$,
  $q^3-q$ and $q^3$ and the following weight enumerator
  \[\begin{array}{l|l}
    \text{Weight} &\multicolumn{1}{c}{ \#\text{ Codewords} } \\ \hline
    q^3-q & q^2(q^2+1)(q-1)/2 \\
    q^3 & q^4-1 \\
    q^3+q & q^2(q^2-1)(q-1)/2
    \end{array} \]
\end{proposition}
\begin{proof}
  The Lagrangian-Grassmmannian $\LL(1,1)$
  is a non--singular hyperplane section of the ordinary
  line--Grassmannian $\GG(3,1)$ of $\PG(3,q)$.
  In particular
  $\LL(1,1)=\GG(3,1)\cap\Sigma=Q(4,q)$,
  where $\Sigma$ is a suitable hyperplane of
  $\PG(4,q)$, depending only on $\sigma$.
  Thus, the code $\cW(2,2)$ is the same as the code determined by the
  projective system of $Q(4,q)$ in a $\PG(4,q)$. Let $\mu$ be the
  orthogonal polarity induced on $\Sigma$ by $Q(4,q)$.
  The three  weights of the code correspond to $3$--spaces which are
  polar (with respect to $\mu$) of points either on $Q(4,q)$ or internal
  or external to it.
  In particular, there are $q^3+q^2+q+1$ points on $Q(4,q)$ and their
  polar hyperplane meets $Q(4,q)$ in a cone consisting of $q^2+q+1$ points;
  the number of internal points is $q^2(q^2-1)/2$ and that of
  external points is $q^2(q^2+1)/2$. Observe that each hyperplane
  corresponds to $(q-1)$ words; this provides the complete enumerator.
\end{proof}

 \section{Lagrangian-Grassmannian codes of rank $3$}
\label{DPS}
In this section we shall provide the full weight enumerator for the
Lagrangian--Grassmannian code $\cW(3,3)$, and discuss some codes
arising from different embeddings of the Symplectic Grassmannian $\Lambda_{3,3}$.



\begin{theorem}
  For $k=n=3$, the minimum distance of the code $\cW(3,3)$ is $q^6-q^4$.
  The enumerator is as follows
  \[\begin{array}{l|l}
    \text{Weight} &\multicolumn{1}{c}{ \#\text{ Codewords} } \\ \hline
    q^6-q^4    & \frac{1}{2}q^2(q^2+1)(q^2+q+1)(q^3+1)(q-1) \\
    q^6        & (q+1)^2(q^2-q+1)(q^2+1)(q^6-q^3+1)(q-1) \\
    q^6+q^3    & q^9(q^4-1)(q-1)  \\
    q^6+q^4    & \frac{1}{2}q^2(q+1)(q^6-1)(q-1) \\
  \end{array} \]
  Furthermore, all codewords of minimum weight lie in the same orbit.
\end{theorem}
\begin{proof}
The theorem is a direct consequence of the classification of the classical
(geometric) hyperplanes of the dual polar space $\Lambda_{3,3}$ of rank $3$
arising from the Grassmann embedding, as provided in \cite{BC,B0}.
We refer, in particular, to \cite[Tables 1,2,3]{BC} for the exact
numbers of hyperplanes and the cardinalities of their intersection.
\end{proof}

\begin{remark}
  We point out that for $q=2^h$ even, the Lagrangian-Grassmannian $\Lambda_{n,n}$
  always affords the spin embedding in $\PG(2^n-1,q)$; see \cite{BCa}.
  In particular,  $\Lambda_{3,3}$ can also be embedded in
  $\PG(7,q)$.
  Such an embedding gives
  rise to a projective system with parameters $N=q^6+q^5+q^4+2q^3+q^2+q+1$
  and $K=8$. The corresponding code has just two weights, see \cite{B0}, namely
  \[\begin{array}[t]{l|l}
    \text{Weight} &\multicolumn{1}{c}{ \#\text{ Codewords} } \\ \hline
    q^6 & (q^2-1)(q^2+1)(q^3+1) \\
    q^6+q^3    & (q^7-q^3)(q-1).
  \end{array}\]
  We observe that for $q=2$, this determines a $[135,8,64]$ code, and
  the best known code with length $N=135$ and dimension $K=8$
  has minimum distance $d=65$ (see \cite{Grassl}).

  Likewise,  $\Lambda_{4,4}$ can be embedded in $\PG(15,q)$ and
  there it also determines
  a $2$--weight code of parameters $[N,K]=[(q^4+1)(q^3+1)(q^2+1)(q+1),16]$
  and weights $q^{10}$ and $q^{10}+q^7$; see \cite{IB}.
\end{remark}

\begin{remark}
For $q=2$, the universal embedding of $\Lambda_{3,3}$ is different
from the Grassmann embedding and it
spans a $\PG(14,2)$; see \cite{bl,li}.
As such it determines a code of length $N=135$, dimension $K=15$ with
weight enumerator as follows, see \cite{B0,P}:
  \[\begin{array}[t]{l|l}
    \text{Weight} &\multicolumn{1}{c}{ \#\text{ Codewords} } \\ \hline
    30 & 36 \\
    48 & 630 \\
    54 & 1120 \\
    62 & 3780 \\
    64 & 7695 \\
  \end{array}\qquad
\begin{array}[t]{l|l}
  \text{Weight} &\multicolumn{1}{c}{ \#\text{ Codewords} } \\ \hline
    70 & 10368 \\
    72 & 7680 \\
    78 & 1080 \\
    80 & 378 \\
  \end{array} \]
  In particular, the minimum distance in this case is $30$.
\end{remark}

\section{Further bounds on the minimum distance}
\label{bounds}
As $\LL(n-1,k-1)$ is a section of $\GG(2n-1,k-1)$ with a subspace
of codimension ${2n}\choose{k-2}$, it is possible to provide a bound
on the minimum distance of $\cW(n,k)$ in terms of higher weights of
the projective Grassmann
code induced by the projective system  $\GG(n-1,k-1)$.
Recall that the $r$--th higher weight of a code $\cC$ induced by
a projective system $\Omega$ consisting of $N$ points is
\[ d_{r}:=N-\max\{ \#(\Omega\cap\Pi) : \Pi \text{ projective subspace of
  codimension $r$ in } \langle\Omega\rangle  \}; \]
see \cite{Wei} for the definition and some properties, as well as~\cite{TV}
for its geometric interpretation;
in the case of Grassmann codes they have been extensively studied in \cite{Gl0,Go,No}.
As $\cW(n,k)$ can be regarded as the intersection of the Grassmannian
$\GG(2n-1,k-1)$ with
a suitable subspace $\Sigma$ of codimension ${{2n}\choose{k-2}}$, we have
\begin{multline*}
  d_{\min}(\cW(n,k))=\#\cW(n,k)-\max\{ \#(\GG(2n-1,k-1)\cap\Pi):
  \Pi\leq\Sigma, \dim(\Sigma/\Pi)=1 \}\geq \\
  \#\cW(n,k)-\max\{\#(\GG(2n-1,k-1)\cap\Pi): \codim_{\bigwedge^k\!V}(\Pi)={{2n}\choose{k-2}}+1\}= \\
  \#\cW(n,k)-\#\GG(2n-1,k-1)+d_{s}, \\
  \end{multline*}
where
$s={{2n}\choose{k-2}}+1$ and
 $d_s$ is the $s$-th higher weight of the Grassmann
code arising from $\cG_{2n,k}$. In general, this bound is not sharp.
This can be seen directly by considering the case of the code
$\cW(n,2)$.
Indeed,  using the
the second highest weight of the Grassmann code, see \cite{No}, we see that
  \begin{multline*}
    d_{\min}(\cW(n,2))\geq \frac{(q^{2n}-1)(q^{2n-2}-1)}{(q-1)(q^2-1)}-
    {2n\brack 2}_q+ q^{2(2n-2)-1}(q+1)= \\
    \frac{(q^{2n}-1)(q^{2n-2}-q^{2n-1})}{(q-1)(q^2-1)}+q^{2(2n-2)-1}(q+1)=
    \frac{q^{4n-2}-2q^{4n-3}+q^{4n-5}+q^{2n-1}-q^{2n-2}}{(q-1)(q^2-1)}\approx \\
      q^{4n-5}-2q^{4n-6}.
    \end{multline*}
This, however, is quite far away from the correct value
for line symplectic Grassmann codes, namely
$d_{\min}(\cW(n,2))=q^{4n-5}-q^{2n-3}$, as we have determined in
Section~\ref{LLG}.

We point out that in~\cite[Proposition 5]{PZ}, an upper  bound on the minimum distance for Lagrangian-Grassmannian codes is given in terms of the dimension of the Lagrangian-Grassmannian variety, that is \[d_{\min}(\cW(n,n))\leq q^{n(n+1)/2}.\]
By Section~\ref{DPS},  we see that this bound is not sharp for $n=2$ and $n=3$.

\end{document}